\documentclass[amsthm]{autart}
\usepackage{graphics} 
\usepackage{epsfig} 
\usepackage{mathptmx} 
\usepackage{times} 
\usepackage{amsthm}
\usepackage{amssymb, amsmath}
\usepackage{epstopdf} 
\usepackage{cite}
\usepackage{thmtools}
\usepackage{etoolbox}
\usepackage{subeqnarray}
\usepackage{cases}
\DeclareMathOperator{\rank}{rank} %
\DeclareMathOperator{\diag}{diag} %
\DeclareMathOperator{\im}{Im} %
\DeclareMathOperator{\re}{Re} %
\newcommand{\norm}[1]{\left\lVert#1\right\rVert}

\theoremstyle{definition}
\newtheorem{real}{Realization}
\begin{document}
\begin{frontmatter}

\title{Cascade and locally dissipative realizations of linear quantum systems for pure Gaussian state covariance assignment\thanksref{footnoteinfo}}
\thanks[footnoteinfo]{
This work was supported by the Australian Research Council and JSPS Grant-in-Aid No. 40513289. The material in this paper
was partially presented at the 2014 IEEE Conference on Control Applications (CCA), Oct. 
8–-10, 2014, Antibes, France. Corresponding author S.~Ma. Tel. +61 2 62688818.}
\author[Australia]{Shan Ma}\ead{shanma.adfa@gmail.com},
\author[Australia]{Matthew J. Woolley}\ead{m.woolley@adfa.edu.au},
\author[Australia]{Ian R. Petersen}\ead{i.r.petersen@gmail.com},
\author[Japan]{Naoki Yamamoto}\ead{yamamoto@appi.keio.ac.jp}
\address[Australia]{School of Engineering and Information Technology, University of New South Wales at the Australian Defence Force Academy, Canberra ACT 2600, Australia}
\address[Japan]{Department of Applied Physics and Physico-Informatics, 
Keio University, Yokohama 223-8522, Japan}

\begin{keyword}
Linear quantum system, Cascade realization, Locally dissipative realization,  Covariance assignment, Pure Gaussian state.
\end{keyword}

\begin{abstract}
This paper presents two realizations of linear quantum systems for covariance 
assignment corresponding to pure Gaussian states. 
The first one is called a cascade realization; given any covariance matrix 
corresponding to a pure Gaussian state, we can construct a cascaded quantum 
system generating that state. 
The second one is called a locally dissipative realization; given a covariance 
matrix corresponding to a pure Gaussian state, if it satisfies certain conditions, 
we can construct a linear quantum system that has only local interactions 
with its environment and achieves the assigned covariance matrix. 
Both realizations are illustrated by examples from quantum optics.  
\end{abstract}
\end{frontmatter}

\section{Introduction}
For stochastic systems, many of the performance objectives are expressed 
in terms of the variances (or covariances) of the system states. 
In a large space structure, for example, the vibration at certain points 
on the structure must be reduced to an acceptable level. 
This objective in fact involves keeping the variances of some variables 
such as deflections within prescribed bounds. 
One way to achieve this is to assign an appropriate matrix value to the 
covariance of the state vector. 
This method, referred to as {\it  covariance assignment}, has been extensively 
studied in a series of papers by Skelton and colleagues, e.g., 
in \cite{HS87:ijc,CS87:tac,SI89:ijc}. 
For linear stochastic systems with white noises, the covariance matrix 
can be computed by solving the Lyapunov equation for the system. 
In this case, the covariance assignment problem reduces to designing 
system matrices such that the corresponding Lyapunov equation has 
a prescribed solution.

Turning our attention to the quantum case, we find that a covariance matrix 
plays an essential role as well in the field of quantum information. 
In particular for a {\it linear quantum system}, the importance of a covariance matrix 
stands out, because it can fully characterize the {\it entanglement} property, 
which is indeed crucial for conducting quantum information processing 
\cite{BP03:book,weedbrook12:rmp}. 
Therefore it should be of great use to investigate the covariance assignment 
problem for linear quantum systems. 
In fact, there are several such proposals; \cite{OHY11:cdc} studies a quantum 
feedback control problem for covariance assignment, and 
\cite{KY12:pra,Y12:ptrsa,IY13:pra,MWPY14:msc} analyze systems that 
generate a {\it pure} Gaussian state. 
Note that, since a Gaussian state (with zero mean) is uniquely determined 
by its covariance matrix, the aforementioned covariance assignment problem is also known as the 
Gaussian state generation problem; 
thus, if a linear quantum system achieves a covariance matrix corresponding 
to a target Gaussian state, we call that the system generates this Gaussian 
state.

Let us especially focus on Refs. \cite{KY12:pra,Y12:ptrsa,IY13:pra,MWPY14:msc}, 
which provide the basis of this paper. 
As mentioned before, in those papers pure Gaussian states 
are examined, which are a particularly important subclass of Gaussian states 
such that the highest performance of Gaussian quantum information processing 
can be realized \cite{BP03:book,weedbrook12:rmp,MFL11:pra,MLGWRN06:prl}. 
Then they provided several methods to construct a stable linear quantum 
system generating a given pure Gaussian state. 
Moreover, conditions for generating an arbitrary pure entangled Gaussian 
state are given there; surely these are important results, because such a state 
serves as an essential resource for Gaussian quantum information processing 
tasks. 
Of course in the literature several methods for generating various pure 
entangled Gaussian states have been proposed. 
For instance, \cite{A06:prl} gives a systematic method to generate an arbitrary 
pure entangled Gaussian state; the idea is to construct a {\it coherent} process 
by applying a sequence of prescribed unitary operations (composed of beam 
splitters and squeezers in optics case) to an initial state. 
Thus this method is essentially a \emph{closed}-system approach. 
In contrast, the approach we take here is an \emph{open}-system one; 
that is, we aim to construct  {\it dissipative} processes such 
that the system is stable and uniquely driven into a desired target pure 
Gaussian state. 
This strategy is categorized into the so-called \emph{reservoir engineering} 
method~\cite{CPBZ93:prl,PCZ96:prl,WC13:prl,KM11:prl,WC14:pra}; 
in general, this approach has a clear advantage that the system has good 
robustness properties with respect to initial states and evolution time.

Now we describe the problem considered in this paper. 
The methods developed in \cite{KY12:pra,Y12:ptrsa,IY13:pra,MWPY14:msc} 
lead to infinitely many linear quantum systems that uniquely generate 
a target pure Gaussian state. 
Some of these systems are easy to implement, while others are not. 
Then a natural question is how to find a linear quantum system that is 
simple to implement, while still uniquely generates the desired pure 
Gaussian state.

In this paper, we provide two convenient realizations of a linear quantum 
system generating a target pure Gaussian state. 
The first one is a {\it cascade realization}, which is a typical system 
structure found in the literature \cite{Gardiner93:prl,N10:tac,P11:auto}. 
We show that, given any covariance matrix corresponding to a pure 
Gaussian state, we can construct a cascaded quantum system uniquely 
generating that state. 
This cascaded system is a series connection of several subsystems in which 
the output of one is fed as the input to the next.
A clear advantage of the cascade realization is that those subsystems can be 
placed at remote sites. 
Note that the cascade structure has also been widely studied in the classical 
control literature \cite{SS90:scl,HJJ05:scl,LH06:tac}.

The second one is a {\it locally dissipative realization}, which is motivated 
by the specific system structure found in, e.g. 
\cite{IY13:pra,KBDKMZ08:pra,RLMM12:pra,TV12:ptrsa}. 
Note that in these references the notion of {\it quasi-locality} has been 
studied, but in this paper we focus on a stronger notion, {\it locality}. 
Here ``locally dissipative'' means that all the system-environment interactions
act only on one system component. 
Implementations of locally dissipative systems should be considerably easier 
than that of systems which have non-local interactions \cite{BR04:book}. 
In this paper, we show that, given a covariance matrix corresponding to a 
pure Gaussian state, if it satisfies certain conditions, we can construct a locally 
dissipative quantum system generating that state. 

Lastly we remark that the state generated by our method is an {\it internal} 
one confined in the system (e.g. an intra-cavity state in optics), rather than 
an external optical field state. 
This means that, if we aim to perform some quantum information processing 
with that Gaussian state, it must be extracted to outside by for instance the 
method developed in \cite{TFSBK14:prl}. 
In particular by acting some non-Gaussian operations such as the cubic-phase 
gate or photon counting on that extracted Gaussian state, we can realize, e.g., 
entanglement distillation and universal quantum computation \cite{weedbrook12:rmp}. 
On the other hand, a generated internal Gaussian state is not necessarily 
extracted to outside for the purpose of precision measurement in the 
scenario of quantum metrology; for instance a spin squeezed state of an 
atomic ensemble can be directly used for ultra-precise magnetometry 
\cite{TA14:jpa}. 

\textit{Notation.} 
For a matrix $A=[A_{jk}]$ whose entries $A_{jk}$ are complex numbers or 
operators, we define $A^{\top}=[A_{kj}]$, $A^{\dagger}=[A_{kj}^{\ast}]$,  
 where the superscript ${}^{\ast}$ denotes either 
the complex conjugate of a complex number or the adjoint of an operator. 
$\diag[\tau_{1},\cdots, \tau_{n}]$ denotes an $n\times n$ diagonal matrix 
with $\tau_{j}$, $j=1,2,\cdots,n$, on its main diagonal. 
$\mathcal{P}_{N}$ is a $2N\times 2N$ permutation matrix defined by 
$\mathcal{P}_{N}[
x_{1} \;x_{2} \;x_{3} \;x_{4} \;\cdots \;x_{2N}]^{\top}=[
x_{1} \;x_{3} \;\cdots \;x_{2N-1} \;x_{2} \;x_{4} \;\cdots \;x_{2N}]^{\top}$ 
for any column vector $[x_{1} \;x_{2}\;x_{3} \;x_{4}  \;\cdots \;x_{2N}]^{\top}$.


\section{Preliminaries} \label{Preliminaries}
We consider a linear quantum system $G$ of $N$ modes. Each mode is characterized by 
a pair of quadrature operators $\{\hat{q}_{j}, \hat{p}_{j}\}$, $j=1,2,\cdots,N$. 
Collecting them into an operator-valued vector 
$\hat{x}\triangleq\left[\hat{q}_{1}\;\cdots\;\hat{q}_{N}\;\; \hat{p}_{1}\;\cdots\;\hat{p}_{N}\right]^{\top}$, 
we write the canonical commutation relations as
\begin{align}
\label{commutation 1}
\left[\hat{x}, \hat{x}^{\top}\right]
  \triangleq\hat{x}\hat{x}^{\top}-\left(\hat{x}\hat{x}^{\top}\right)^{\top}
  =i\Sigma, \quad \Sigma\triangleq\begin{bmatrix}
         0 & I_{N}\\
-I_{N} &0
\end{bmatrix}. 
\end{align}
Here we emphasize that the transpose operation $\top$, when applied to 
an operator-valued matrix (say, $\hat{x}\hat{x}^{\top}$), only exchanges 
the indices of the matrix and leaves the entries unchanged. 
Therefore $\left(\hat{x}\hat{x}^{\top}\right)^{\top}\ne \hat{x}\hat{x}^{\top}$. Let $\hat{H}$ be the Hamiltonian of the system,  and let 
$\{\hat{c}_{j}\}$, $j=1,2,\cdots,K$, be Lindblad operators that represent the interactions between the system and its environment. For convenience, we collect all the Lindblad operators as an operator-valued vector $\hat{L}=\begin{bmatrix}
\hat{c}_{1} &\hat{c}_{2} &\cdots & \hat{c}_{K}
\end{bmatrix}^{\top}$ and
call $\hat{L}$ the \emph{coupling vector}.
Suppose $\hat{H}$ is quadratic in $\hat{x}$, i.e., 
$\hat{H}=\frac{1}{2}\hat{x}^{\top}M\hat{x}$,
with $M=M^{\top}\in \mathbb{R}^{2N \times 2N}$, and $\hat{L}$ is linear in $\hat{x}$, i.e., 
$\hat{L} = C \hat{x}$, 
with $C\in \mathbb{C}^{K \times 2N}$, then the quantum system $G$ can be described by the following quantum stochastic differential equations (QSDEs)
\begin{equation} \label{QSDEequation}
\left\{\begin{aligned}
d\hat{x}(t)&=\mathcal{A}\hat{x}(t)dt+\mathcal{B} \begin{bmatrix}
d\hat{A}^{\top}(t) &d\hat{A}^{\dagger}(t)
\end{bmatrix}^{\top},\\
d\hat{Y}(t)&=\mathcal{C}\hat{x}(t)dt+d\hat{A}(t),
 \end{aligned}\right.
\end{equation}
where $\mathcal{A}=\Sigma(M+\im(C^{\dagger}C))$, 
$\mathcal{B}=i\Sigma[- C^{\dagger}\; \; C^{\top}]$, $\mathcal{C}=C$~
\cite{Y12:ptrsa},  \cite[Chapter~6]{WM10:book}. The input $d\hat{A}(t) =\left[d\hat{A}_{1}(t)\;\; \cdots \;\; d\hat{A}_{K}(t)\right]^{\top}$ represents $K$ independent quantum stochastic processes,  
with $d\hat{A}_{j}(t)$, $j=1,2,\cdots,K$, satisfying the following quantum It\=o rules: 
\begin{equation} \label{ito}
\left\{\begin{aligned} 
d\hat{A}_{j}(t)d\hat{A}_{k}^{\ast}(t)&=\delta_{jk}dt,\\
 d\hat{A}_{j}(t)d\hat{A}_{k}(t)&
   =d\hat{A}_{j}^{\ast}(t)d\hat{A}_{k}^{\ast}(t)=d\hat{A}_{j}^{\ast}(t)d\hat{A}_{k}(t)=0,
 \end{aligned}\right. 
\end{equation}
where $\delta_{jk}$ is the Kronecker $\delta$-function. 
The output $d\hat{Y}(t)=\left[d\hat{Y}_{1}(t)\;\; \cdots \;\; d\hat{Y}_{K}(t)\right]^{\top}$ 
satisfies quantum It\=o rules similar to~\eqref{ito}~\cite{HP84:cmp,B92:JMA,GZ00:book, BHJ07:siamjco,WM10:book,Y12:ptrsa}. 
The quantum expectation of the  
vector $\hat{x} $ is denoted by  
$\langle \hat{x} \rangle $ 
and the covariance matrix  is given by   $V=\frac{1}{2}\langle \triangle\hat{x}{\triangle\hat{x}}^{\top}+(\triangle\hat{x}{\triangle\hat{x}}^{\top})^{\top} \rangle$, where $\triangle\hat{x}=\hat{x}-\langle \hat{x}\rangle$; see, e.g.,~\cite{KY12:pra,Y12:ptrsa,MFL11:pra}.  The time evolutions of the mean vector $\langle \hat{x}(t) \rangle$ 
and the covariance matrix    $V(t)$ can be derived from~\eqref{QSDEequation} 
by using the quantum It\=o rule. They are given by
\begin{numcases}{}
\frac{d\langle\hat{x}(t)\rangle}{dt}=\mathcal{A}\langle\hat{x}(t)\rangle,\label{meanfunction2} \\
\frac{dV(t)}{dt}=\mathcal{A}V(t)+V(t)\mathcal{A}^{\top}+\frac{1}{2}\mathcal{B}\mathcal{B}^{\dagger}. \label{covfunction2}
\end{numcases}
As in the classical case, a Gaussian state is completely characterized by the  mean vector $\langle \hat{x}\rangle$ 
and the covariance matrix   $V$. Since the mean vector $\langle \hat{x}\rangle$ 
contains no information about noise and entanglement, 
we will restrict our attention to zero-mean Gaussian states (i.e., $\langle \hat{x}\rangle=0$). 
  A Gaussian state is pure if and only if its  covariance matrix   $V$ satisfies $\det(V)=2^{-2N}$. 
In fact, when a Gaussian state is pure, its covariance matrix   $V$ can always be factored as
\begin{align}\label{covariance}
V=\frac{1}{2}SS^{\top},\quad S=\begin{bmatrix}
Y^{-\frac{1}{2}} &0\\
XY^{-\frac{1}{2}} &Y^{\frac{1}{2}} 
\end{bmatrix},
\end{align}
where $X=X^{\top}\in \mathbb{R}^{N \times N}$, $Y=Y^{\top}\in \mathbb{R}^{N \times N}$ and $Y>0$~\cite{MFL11:pra,WGKWC04:pra}. 
For example, the $N$-mode vacuum state is a special pure Gaussian state with $X=0$ and $Y=I_{N}$. It can be seen from~\eqref{covariance}  that a pure Gaussian state is uniquely specified by a complex, symmetric matrix $Z\triangleq X+iY$, which is referred to as the \emph{graph} corresponding to a pure Gaussian state~\cite{MFL11:pra}. 
Note also that the matrix $S$ satisfies $S\Sigma S^{\top}=\Sigma$, 
which means that $S$ is a symplectic matrix. The symplectic nature of $S$ guarantees that 
the mapping $\hat{x}\longmapsto \hat{x}'\triangleq S\hat{x}$ 
preserves the canonical commutation relations~\eqref{commutation 1}, that is
\begin{align*}
\left[\hat{x}', \hat{x}'^{\top}\right]=\left[S\hat{x}, \left(S\hat{x}\right)^{\top}\right]=S\left[\hat{x}, \hat{x}^{\top}\right]S^{\top}=S\left(i\Sigma\right) S^{\top}=i\Sigma.
\end{align*}
Note that if the  system $G$ is initially in a Gaussian state, 
then the system $G$ will always be Gaussian, with the mean vector $\langle \hat{x}(t) \rangle$ 
and the covariance matrix   $V(t)$ obeying~\eqref{meanfunction2} and~\eqref{covfunction2}, respectively. We shall be particularly interested in the steady-state covariance matrix $V(\infty)$.

Assume that the system $G$ is initially in a Gaussian state. 
The problem of pure Gaussian state covariance assignment is to find a Hamiltonian  $\hat{H}$ 
and a coupling vector $\hat{L}$ such that the corresponding linear quantum system described by~\eqref{QSDEequation} is asymptotically stable and achieves the covariance matrix  corresponding to a given pure Gaussian state. Since a pure Gaussian state (with zero mean) is uniquely specified by its covariance matrix, so if a linear quantum system achieves a  covariance matrix corresponding to a pure Gaussian state,  we can simply say that such a linear quantum system uniquely generates the  pure Gaussian state. The problem can be expressed mathematically as: 
\begin{align*}
&\text{find}\quad\quad\quad\quad  M=M^{\top}\in \mathbb{R}^{2N \times 2N}\;\; \text{and}\;\; C\in \mathbb{C}^{K \times 2N}\\
&\text{subject to}\quad\;\;\; \mathcal{A}\;\; \text{is Hurwitz},\\
&\quad\quad\quad\quad\quad \;\;                \mathcal{A}V+V\mathcal{A}^{\top}+\frac{1}{2}\mathcal{B}\mathcal{B}^{\dagger}=0,
\end{align*}
where $V$ is the covariance matrix  corresponding to the desired target pure Gaussian state. Here a matrix  $\mathcal{A}$ is said to be \emph{Hurwitz} if all its eigenvalues have strictly negative real parts. A system described by~\eqref{QSDEequation} is said to be \emph{asymptotically stable} if the  matrix
$\mathcal{A}$ is a Hurwitz matrix.
 Recently, a necessary and sufficient condition has been developed in~\cite{KY12:pra} for solving the pure Gaussian state covariance assignment problem. 
 The result is summarized as follows. 
\begin{lem}[\cite{KY12:pra,Y12:ptrsa}]\label{lem1}
Let $V$ be the covariance matrix  corresponding to  a given $N$-mode pure Gaussian state. 
Assume that $V$ is expressed in the factored form~\eqref{covariance}.   
Then this pure Gaussian state is uniquely generated by the linear quantum system~\eqref{QSDEequation} if and only if 
\begin{align} \label{G}
M=\begin{bmatrix}
XRX+YRY-\Gamma Y^{-1}X-XY^{-1}\Gamma^{\top} &-XR+\Gamma Y^{-1}\\
-RX+Y^{-1}\Gamma^{\top} &R
\end{bmatrix},
\end{align}
 and 
\begin{align} \label{C}
C=P^{\top}[-Z \;\; I_{N}], 
\end{align} 
where $R=R^{\top}\in\mathbb{R}^{ N\times N}$, $\Gamma=-\Gamma^{\top}\in\mathbb{R}^{ N\times N}$, 
and $P\in \mathbb{C}^{N\times K}$ are free matrices satisfying the following rank condition
\begin{align}\label{rankconstraint}
        \rank\left(\begin{bmatrix}P &QP &\cdots &Q^{N-1}P\end{bmatrix}\right)=N,
           \;\;  Q\triangleq -iRY+Y^{-1}\Gamma. 
\end{align}
\end{lem}

\begin{rem}
From~\eqref{C}, we see that the resulting coupling vector $\hat{L}$ of the engineered system is $
\hat{L}=C\hat{x}=P^{\top}[-Z \;  I_{N}]\hat{x}=P^{\top} (\left[\hat{p}_{1}\;\cdots\;\hat{p}_{N}\right]^{\top}-Z\left[\hat{q}_{1}\;\cdots\;\hat{q}_{N}\right]^{\top} )$.
Therefore, all the components of $\hat{L}$ are \emph{nullifiers} for the desired target pure Gaussian state~\cite{MFL11:pra}. As a special example, one can  engineer a purely dissipative system (with $\hat{H}=0$) to generate a pure Gaussian state. In this case, one could take $R=\Gamma=0_{N\times N}$ and $P=I_{N}$ in Lemma~\ref{lem1}. Then the resulting coupling vector $\hat{L}$ is the so-called \emph{nullifier vector} for the desired target pure Gaussian state. 
\end{rem}

\begin{rem}\label{rmk2}
 Lemma~\ref{lem1}  has a simple interpretation in terms of symplectic transformations~\cite{TN15:scl,Y12:ptrsa}. 
 As mentioned before, vacuum states are a special class of pure Gaussian states. 
 The covariance matrix  corresponding to the $N$-mode vacuum state is $V=\frac{1}{2}I_{N}$. 
By using physical realizability conditions, it can be proved that the $N$-mode vacuum state can only be generated by an $N$-mode passive linear quantum system~\cite{TN15:scl}. 
The converse is also true. That is, an $N$-mode passive linear quantum system, 
if it is asymptotically stable, must evolve toward the $N$-mode vacuum state~\cite{H14:tac}.   
 Recall that for a passive linear quantum system,  
the Hamiltonian is always  of the form $\hat{H}=\frac{1}{2}\hat{x}^{\top}\tilde{M}\hat{x}$, with $\tilde{M}=\begin{bmatrix}
\tilde{R} &\tilde{\Gamma}\\
\tilde{\Gamma}^{\top} &\tilde{R}
\end{bmatrix}$, $\tilde{R}=\tilde{R}^{\top}\in\mathbb{R}^{ N\times N}$, 
and  $\tilde{\Gamma}=-\tilde{\Gamma}^{\top}\in\mathbb{R}^{ N\times N}$, 
and the coupling vector is always of the form $\hat{L} = \tilde{C}\hat{x}$, 
with $\tilde{C}=\tilde{P}^{\top}\left[-iI_{N}\;\; I_{N}\right]$, $\tilde{P}\in \mathbb{C}^{N\times K}$~\cite{P11:auto,P12:scl,GY16:tac,GZ15:auto}. Now we apply a symplectic transformation to $\hat{x}$, 
that is, we define $\hat{x}'\triangleq S\hat{x}$. Then, in terms of $\hat{x}'$, 
the Hamiltonian is rewritten as $\hat{H}=\frac{1}{2}\hat{x}'^{\top}S^{-\top}\tilde{M}S^{-1}\hat{x}'$ 
and  the coupling vector is rewritten as $\hat{L} = \tilde{C} S^{-1}\hat{x}'$. 
We also observe that the relation between the covariance matrix $V'$ of $\hat{x}'$ 
and the covariance matrix $V$ of $\hat{x}$ is given as follows:
\begin{align*}
V'&=\frac{1}{2}\langle \triangle\hat{x}'{\triangle\hat{x}}'^{\top}+(\triangle\hat{x}'{\triangle\hat{x}}'^{\top})^{\top} \rangle\\
&=\frac{1}{2}S\langle \triangle\hat{x}{\triangle\hat{x}}^{\top}+(\triangle\hat{x}{\triangle\hat{x}}^{\top})^{\top} \rangle S^{\top}\\
&=SV S^{\top}.
\end{align*}
If the passive linear quantum system is asymptotically stable, then based on the result in~\cite{H14:tac}, we have  
$V\to \frac{1}{2}I_{N},\; \text{as}\; \; t \to +\infty$.
As a result, $V'\to \frac{1}{2} SS^{\top}$, which gives the desired pure Gaussian state. 
Combining the results above, we conclude that for a given pure Gaussian state $V=\frac{1}{2} SS^{\top}$, 
a complete parametrization of the linear quantum system $G$ that uniquely generates this pure Gaussian state is given by
\begin{numcases}{}
M=S^{-\top}\tilde{M}S^{-1},\label{Mform}\\
C=\tilde{C} S^{-1}, \label{Cform}
\end{numcases}
where $(\tilde{M}, \tilde{C})$ form an asymptotically stable passive linear quantum system. 
Substituting $\tilde{M}=\begin{bmatrix}
\tilde{R} &\tilde{\Gamma}\\
\tilde{\Gamma}^{\top} &\tilde{R}
\end{bmatrix}$ and $\tilde{C}=\tilde{P}^{\top}\left[-iI_{N},\; I_{N}\right]$ into~\eqref{Mform},~\eqref{Cform} 
and using some additional matrix transformations, we will obtain the formulas~\eqref{G},~\eqref{C}, respectively.  
This is the idea behind Lemma~\ref{lem1}. The rank constraint~\eqref{rankconstraint} indeed 
gives a sufficient and necessary stability condition 
for the original passive linear quantum system~\cite{GY16:tac,GZ15:auto}. 
As a result, it also guarantees the stability of the linear quantum system $G$ 
 based on the linear transformation theory in the control field~\cite{ZJK96:book}. 
\end{rem}


\section{The cascade realization} \label{The cascade realization}
As we have seen in Lemma~\ref{lem1}, the matrices $R$, $\Gamma$ and $P$ 
are free matrices, although they must satisfy the rank condition~\eqref{rankconstraint}. 
By varying them we can obtain different linear quantum systems  that uniquely generate 
a given pure Gaussian state. 
Based on this fact, we provide two feasible realizations of linear quantum systems for covariance assignment corresponding to pure Gaussian states, and this section is devoted to 
 the first one, the cascade realization. 


\subsection{The cascade realization}
For convenience, we denote a linear quantum system $G$ with the Hamiltonian $\hat{H}$ and the coupling vector $\hat{L}$
 as $G = (\hat{H},\; \hat{L})$. 
Suppose we have two linear quantum systems 
$G_{1}= (\hat{H}_{1},\; \hat{L}_{1})$ and $G_{2}= (\hat{H}_{2},\;\hat{L}_{2})$. 
If we feed the output of the system $G_{1}$  into the input of the system $G_{2}$, 
we will obtain a cascaded quantum system $G= G_{2}\lhd G_{1}$, as shown in 
Fig.~\ref{twocascadestructure}. 
\begin{figure}[htbp]
\begin{center}
\includegraphics[height=1.05cm]{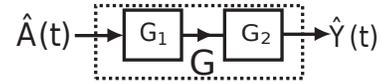}
\caption{The cascade connection of two linear quantum systems: $G= G_{2}\lhd G_{1}$. }
\label{twocascadestructure}
\end{center}
\end{figure}
Based on the quantum theory of cascaded linear quantum systems~\cite{GJ09:tac}, the 
Hamiltonian $\hat{H}$ and the coupling vector  $\hat{L}$ of the  cascaded 
system $G$ are, respectively, given by
\begin{equation}
\left\{\begin{aligned}\label{SLHformula1}
\hat{H}&=\hat{H}_{2}+\hat{H}_{1}+\frac{1}{2i}\left(\hat{L}_{2}^{\dagger}\hat{L}_{1}-\hat{L}_{1}^{\dagger}\hat{L}_{2}\right),\\
\hat{L}&=\hat{L}_{2}+\hat{L}_{1}.
 \end{aligned}\right. 
\end{equation}
This result can be extended to the cascade connection of $N$ one-dimensional 
 harmonic oscillators. 
Suppose we have $N$ one-dimensional harmonic oscillators $G_{j}$ 
with the Hamiltonian $\hat{H}_{j}=\frac{1}{2}\hat{\xi}_{j}^{\top}M_{j}\hat{\xi}_{j}$, 
$M_{j}=M_{j}^{\top}\in \mathbb{R}^{2 \times 2}$, 
$\hat{\xi}_{j}\triangleq[\hat{q}_{j}\;\hat{p}_{j}]^{\top}$, and the coupling 
vector  $\hat{L}_{j} = C_{j} \hat{\xi}_{j}$, $C_{j}\in \mathbb{C}^{K \times 2}$, 
$j=1,2,\cdots,N$. 
The system $G$ is obtained by a cascade connection of these 
harmonic oscillators, that is, $G=G_{N}\lhd \cdots \lhd G_{2}\lhd G_{1}$, as 
shown in Fig.~\ref{ncascadestructure}. By repeatedly using~\eqref{SLHformula1}, 
the Hamiltonian $\hat{H}$ and 
the coupling vector  $\hat{L}$ of the cascaded system $G$ are given by the following lemma. 
\begin{figure}[htbp]
\begin{center}
\includegraphics[height=1cm]{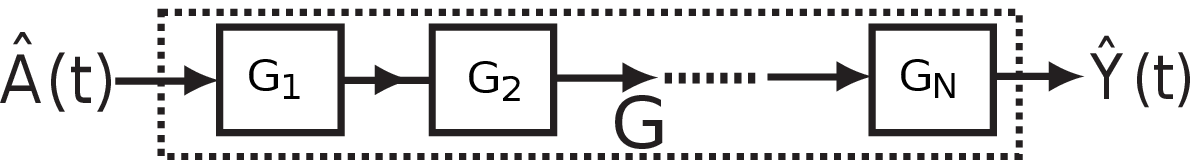}
\caption{The cascade connection of $N$ one-dimensional  harmonic oscillators: $G=G_{N}\lhd \cdots \lhd G_{2}\lhd G_{1}$.}
\label{ncascadestructure}
\end{center}
\end{figure}
\begin{lem}[\protect{\cite{N10:tac}}] \label{nurdin}
Suppose that the system $G$ is obtained via a cascade connection of the
aforementioned $N$ one-dimensional 
 harmonic oscillators $G_{j}$, $j=1,2,\cdots,N$, that is, 
$G=G_{N}\lhd \cdots \lhd G_{2}\lhd G_{1}$. 
Then the Hamiltonian $\hat{H}$ and the coupling vector   $\hat{L}$ of the 
 system $G$ are, respectively, given by 
\begin{equation}
\left\{\begin{aligned}
\hat{H}&=\frac{1}{2}\hat{x} ^{\top}M\hat{x},\quad M=\mathcal{P}_{N}\mathbb{M} \mathcal{P}_{N}^{\top}\\
\hat{L}&=C\hat{x}, \quad C=\left[
C_{1} \;\; C_{2} \;\;\cdots \;\;C_{N}
\right]\mathcal{P}_{N}^{\top},
\end{aligned}\right. \notag
\end{equation}
where $\mathbb{M} =[\mathbb{M} _{jk}]_{j,k=1,\cdots,N}$ is a symmetric block matrix with 
$\mathbb{M} _{jj}=M_{j}$, $\mathbb{M} _{jk}=\im (C_{j}^{\dagger}C_{k})$ 
whenever $j>k$ and $\mathbb{M} _{jk}=\mathbb{M} _{kj}^{\top}$ whenever 
$j<k$. 
\end{lem}

It can be seen  from Lemma~\ref{nurdin} that due to the cascade feature, the Hamiltonian 
matrix $M$ and the coupling matrix $C$ of the cascaded system $G$ depend on each other 
in a complicated way. 
Nevertheless, given any pure Gaussian state, we can always construct a cascade 
connection of several one-dimensional  harmonic oscillators such 
that this cascaded quantum system is asymptotically stable and achieves the covariance matrix  corresponding to the desired target pure Gaussian state.  
The result is stated as follows. 

\begin{thm} \label{theorem1}
Any $N$-mode pure Gaussian state can be uniquely generated by constructing a cascade of $N$ one-dimensional harmonic oscillators. 
\end{thm}
\begin{proof}
We prove this result by construction. 
Recall that for an arbitrary $N$-mode pure Gaussian state, the corresponding covariance matrix  $V$ has the factorization 
shown in~\eqref{covariance}. 
Using the matrices $X$ and $Y$ obtained from~\eqref{covariance}, we construct 
a cascaded system $G=G_{N}\lhd\cdots\lhd G_{2}\lhd G_{1}$ with the Hamiltonian 
$\hat{H}_{j}$ and the coupling vector  $\hat{L}_{j}$, $j=1,2,\cdots,N$, given by 
\begin{equation} \notag
\left\{\begin{aligned}
\hat{H}_{j}&=0,\\
\hat{L}_{j}&=C_{j}\hat{\xi}_{j},\; C_{j}
=iY^{-\frac{1}{2}}\left[-Z\;I_{N}\right]\mathcal{P}_{N}\begin{bmatrix}
0_{(2j-2)\times 2}\\
I_{2}\\
0_{(2N-2j)\times 2}
\end{bmatrix}.
 \end{aligned}\right.
\end{equation}
 Using Lemma~\ref{nurdin}, we can calculate the Hamiltonian 
$\hat{H}=\frac{1}{2}\hat{x}^{\top}M\hat{x}$ and the coupling vector 
$\hat{L}=C\hat{x}$ for the cascaded system $G$. We find that
$M=0$ and $C=iY^{-1/2}\left[-Z\;\;I_{N}\right]$. 
Then it follows from the QSDE~\eqref{QSDEequation} that
\begin{align*}
\mathcal{A}&=\Sigma(M+\im(C^{\dagger}C))\\
&=\Sigma \im \left(\begin{bmatrix}
(X-iY)Y^{-1}(X+iY) &-(X-iY)Y^{-1}\\
-Y^{-1}(X+iY) & Y^{-1}
\end{bmatrix}\right)\\
&=\Sigma\Sigma =-I_{2N},\\
\mathcal{D}&\triangleq\frac{1}{2}\mathcal{B}\mathcal{B}^{\dagger}=\Sigma\re(C^{\dagger}C)\Sigma^{\top}\\
&=\Sigma\re \left(\begin{bmatrix}
(X-iY)Y^{-1}(X+iY) &-(X-iY)Y^{-1}\\
-Y^{-1}(X+iY) & Y^{-1}
\end{bmatrix}\right) \Sigma^{\top} \\
&=\begin{bmatrix}
Y^{-1} &Y^{-1}X\\
XY^{-1} &XY^{-1}X+Y
\end{bmatrix}.
\end{align*}
Clearly, $\mathcal{A}$ is Hurwitz. Furthermore, it can be verified that 
\begin{align}
&\mathcal{A}V +V \mathcal{A}^{\top}+\mathcal{D}=0.\label{lyapunov2}
\end{align}  

The  stability of $\mathcal{A}$ and the Lyapunov 
equation~\eqref{lyapunov2} guarantee that  the cascaded system $G$ constructed above is asymptotically stable 
and  achieves the covariance matrix $V$. In other words, the cascaded system $G$ uniquely generates  the desired target pure Gaussian state.
\hfill  
\end{proof}

\subsection{Example} 
\begin{exmp} \label{exam1}
We consider the generation of {\it two-mode squeezed states}~\cite{MFL11:pra}. 
Two-mode squeezed states are highly symmetric entangled states, which are 
very useful in several quantum information protocols such as quantum 
teleportation~\cite{IYYYF12:tac}. 
The covariance matrix  $V$ corresponding to a two-mode squeezed state is
\begin{align} \label{twomode-covariance}
V=\frac{1}{2}\begin{bmatrix}
\cosh(2\alpha) &\sinh(2\alpha) &0 &0\\
\sinh(2\alpha) &\cosh(2\alpha) &0 &0\\
0 &0 &\cosh(2\alpha) &-\sinh(2\alpha)\\
0 &0 &-\sinh(2\alpha) &\cosh(2\alpha)
\end{bmatrix},
\end{align}
where $\alpha$ is the squeezing parameter. 
Using the factorization~\eqref{covariance}, we have $
X=0$ and $Y=\begin{bmatrix}
\cosh(2\alpha) &-\sinh(2\alpha)\\
 -\sinh(2\alpha) &\cosh(2\alpha)
\end{bmatrix}$. Therefore, the graph corresponding to a two-mode squeezed state is given by $Z=X+iY=\begin{bmatrix}
i\cosh(2\alpha) &-i\sinh(2\alpha)\\
 -i\sinh(2\alpha) &i\cosh(2\alpha)
\end{bmatrix}$.

Next we provide two different cascade realizations. 
The first one, Realization~1, is constructed based on a heuristic derivation, while   
the second one,  Realization~2, is constructed 
based on the proof of Theorem~\ref{theorem1}.  
\begin{real}
In this cascade realization, the subsystems 
$G_{1}=(\hat{H}_{1},\;\hat{L}_{1})$ and $G_{2}=(\hat{H}_{2},\;\hat{L}_{2})$ are, respectively,  given by 
\begin{equation}  
\left\{\begin{aligned} 
\hat{H}_{1}&=\frac{1}{2}\hat{\xi}_{1}^{\top}\begin{bmatrix}
2 &Q_{1}\\
Q_{1} &2
\end{bmatrix}\hat{\xi}_{1},\quad \; \hat{L}_{1}=[iQ_{2}\; 1 ]\hat{\xi}_{1}, \\
\hat{H}_{2}&=-\frac{1}{2}\hat{\xi}_{2}^{\top}\begin{bmatrix}
2 &Q_{1}\\
Q_{1} &2
\end{bmatrix}\hat{\xi}_{2},\;\;\;   \hat{L}_{2}=[iQ_{2}\; 1 ]\hat{\xi}_{2}, 
 \end{aligned}\right. \notag
\end{equation}
where $Q_{1}\triangleq\frac{\sinh^{2}(2\alpha)}{\cosh(2\alpha)}-\sinh(2\alpha)$ and  $Q_{2}\triangleq\sinh(2\alpha)-\cosh(2\alpha)$.
It can be proved that the cascaded system $G=G_{2}\lhd G_{1}$ is asymptotically stable 
and achieves the covariance matrix~\eqref{twomode-covariance}. 
The proof is similar to that of Theorem~\ref{theorem1}, and hence is omitted. 
Using the result in~\cite{NJD09:siamjco}, a corresponding quantum optical realization is provided in Fig.~\ref{two-mode-realization1}. For each subsystem $G_{j}$, $j=1,\;2$, the Hamiltonian $\hat{H}_{j}$ is realized by a nonlinear crystal pumped by a classical field, and 
the coupling operator $\hat{L}_{j}$ is realized by implementing an auxiliary cavity. 
This auxiliary cavity interacts with the subsystem via a cascade of a pumped 
 crystal and a beam splitter. 
It has a fast mode that can be adiabatically eliminated.  
\end{real}
\begin{figure}[htbp]
\begin{center}
\includegraphics[height=3.2cm]{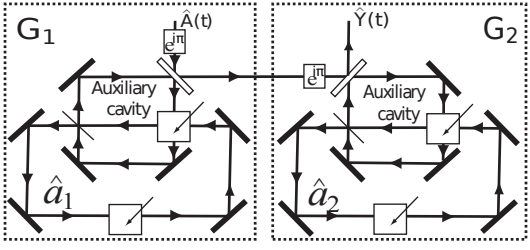}
\caption{An optical cascade realization of the two-mode linear quantum system 
that uniquely generates a two-mode squeezed state.  
The square with an arrow represents a pumped  crystal.  
The symbol $e^{i\pi}$ with a square on it represents a phase shift $\pi$. Solid (dark) rectangles 
denote perfectly reflecting mirrors, while unfilled rectangles denote partially 
transmitting mirrors. 
The dark line ``\textbf{$\mathbb{\diagdown}$}'' represents an optical beam 
splitter.
}
\label{two-mode-realization1}
\end{center}
\end{figure}

\begin{real}
The second realization is constructed according to the method shown in the 
proof of Theorem~\ref{theorem1}. 
By direct calculation, the subsystems 
$G_{1}=(\hat{H}_{1},\;\hat{L}_{1})$ and $G_{2}=(\hat{H}_{2},\;\hat{L}_{2})$ are, respectively,  given by  
\begin{equation}  
\left\{\begin{aligned} 
\hat{H}_{1}&=0,\quad \hat{L}_{1}=\begin{bmatrix}
\cosh(\alpha) &i\cosh(\alpha)\\
-\sinh(\alpha) &i\sinh(\alpha)
\end{bmatrix}\hat{\xi}_{1},\\
 \hat{H}_{2}&=0,\quad \hat{L}_{2}=\begin{bmatrix}
-\sinh(\alpha) &i\sinh(\alpha)\\
\cosh(\alpha) &i\cosh(\alpha)
\end{bmatrix}\hat{\xi}_{2}. 
 \end{aligned}\right. \notag
\end{equation}

Using the result in~\cite{NJD09:siamjco}, a corresponding quantum optical realization of 
such a cascaded quantum system $G=G_{2}\lhd G_{1}$ is provided in 
Fig.~\ref{two-mode-realization2}. 
This cascaded system $G$ has two crucial features. 
First, because $\hat{H}_{1}=\hat{H}_{2}=0$, 
implementations of the Hamiltonians involve no pumped  crystals. 
Second, the first component of the coupling vector  
$\hat{L}_{1}=[\hat{c}_{1,1}\; \hat{c}_{1,2}]^\top$ is
$
\hat{c}_{1,1}= \begin{bmatrix}
\cosh(\alpha) &i\cosh(\alpha)
\end{bmatrix}\begin{bmatrix}
\hat{q}_{1}\\
\hat{p}_{1}
\end{bmatrix}=\sqrt{2}\cosh(\alpha)\hat{a}_{1}$,
where $\hat{a}_{1}=(\hat{q}_1+i\hat{p}_1)/\sqrt{2}$ denotes the annihilation 
operator of the first mode. 
This operator $\hat{c}_{1,1}$ represents the standard linear dissipation of a cavity mode 
into a continuum of field modes outside of the cavity. 
A similar case also occurs in the coupling vector  $\hat{L}_{2}$. 
As can be seen in Fig.~\ref{two-mode-realization2}, Realization~2 requires two pumped 
 crystals, in contrast to the case of Realization~1, where four pumped  crystals are used. 
From this viewpoint,  Realization~2, which is constructed based on our result, has a clear advantage over Realization~1. 

\begin{figure}[htbp]
\begin{center}
\includegraphics[height=3.85cm]{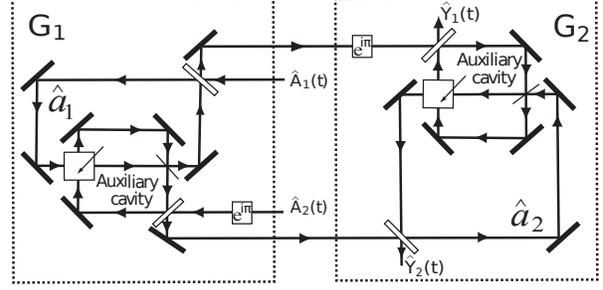}
\caption{
Another optical cascade realization of the two-mode linear quantum system 
that uniquely generates a two-mode squeezed state. }
\label{two-mode-realization2}
\end{center}
\end{figure}
\end{real}
\end{exmp}


\section{The locally dissipative  realization} \label{The locally dissipative  realization}
In this section, we describe the second realization of linear quantum systems for covariance assignment corresponding to  pure Gaussian states.  Unlike the cascade realization, the  locally dissipative realization  cannot generate all pure Gaussian states, but as shown later the class of stabilizable states is fairly broad.

\subsection{The locally dissipative  realization}
As we have noted in Section~\ref{Preliminaries}, the coupling vector   $\hat{L}$ is an operator-valued vector that consists of $K$ elements, i.e., 
$\hat{L}=\left[
\hat{c}_{1} \;\hat{c}_{2} \;\cdots \;\hat{c}_{K}
\right]^{\top}$. Each element $\hat{c}_{j}$, $j=1,2,\cdots,K$, called a Lindblad operator, represents
an interaction between the system and its environment. 
A Lindblad operator $\hat{c}_{j}$ is said to be local if it acts only on one system mode. 
As an example, consider the system depicted in~Fig.~\ref{local}. 
The Lindblad operator $\hat{c}_{1}=\hat{q}_{1}+\hat{p}_{1}$ acts only on the first system mode, so it is a local 
operator. 
On the other hand, the Lindblad operator $\hat{c}_{2}=\hat{q}_{1}+\hat{q}_{2}$ 
acts on two system modes, so by definition it is not a  local operator.  If all the Lindblad operators in $\hat{L}$ 
are local, then the system
is called a locally dissipative quantum system. A locally dissipative quantum system could be relatively easy to
implement in practice. Therefore, we would like to characterize the class of pure Gaussian states that 
can be generated using locally dissipative quantum systems.
The result is given by the following theorem. 
\begin{figure}[htbp]
\begin{center}
\includegraphics[height=1.2cm]{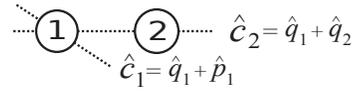}
\caption{An illustration of local Lindblad operators. 
$\hat{c}_{1}$ is a local Lindblad operator, while $\hat{c}_{2}$ is not a local one.}
\label{local}
\end{center}
\end{figure}

\begin{thm}\label{thm2}
Let $V$ be the covariance matrix  corresponding to  a given $N$-mode pure Gaussian state. 
Assume that it is expressed in the factored form~\eqref{covariance}. 
Then this pure Gaussian state can be uniquely generated in an $N$-mode locally dissipative quantum system 
if and only if there exists an integer $\ell$, $1\le \ell\le N$, such that 
\begin{align}
Z_{(\ell, j)}=Z_{(j,\ell )}=0,\quad \forall j\ne \ell \;\; \text{and}\;\; 1\le j\le N, \label{thm21}
\end{align}
where $Z_{(\ell, j)}$ denotes the $(\ell,j)$ element of the graph matrix $Z=X+iY$ for the  pure Gaussian state.
\end{thm}
  
\begin{proof}
We prove the sufficiency part by construction. 
Equation~\eqref{thm21} implies that there exists a row
vector  
$\Upsilon= \begin{bmatrix}
 0_{1\times (\ell-1)}\;\;\tau_{1}\;\; 0_{1\times (N-\ell)}
 \end{bmatrix}$ with $\tau_{1}\ne 0$, such that
$
 \Upsilon Z = \begin{bmatrix}
 0_{1\times (\ell-1)}\;\;\tau_{2} \;\; 0_{1\times (N-\ell)}
\end{bmatrix}$, 
where $\tau_{2}=\tau_{1}Z_{(\ell, \ell)}$. 
Using the Gram-Schmidt method, we can create three $N\times N$ matrices 
$U_{1}$, $U_{2}$ and $\Lambda$, where $U_{1}$ is a unitary matrix with 
the first column being $\frac{Y^{1/2}\Upsilon^{\top}}{\norm{Y^{1/2}\Upsilon^{\top}}}$; 
$U_{2}$ is a unitary matrix with the first column being $\frac{1}{\sqrt{N}} 
\begin{bmatrix} 1 &1 &\cdots &1
\end{bmatrix}^{\top}$ and $\Lambda$ is a purely imaginary matrix 
$\Lambda=i\diag[\alpha_{1},\cdots, \alpha_{N}]$, with $\alpha_{j}\in \mathbb{R}$, 
$j=1,\cdots, N$, and $\alpha_{j} \ne \alpha_{k}$,  $\forall j\ne k$.

Let $P=\Upsilon^{\top}$, $R=-Y^{-1/2}\im(U_{1}U_{2}^{\dagger}\Lambda U_{2}U_{1}^{\dagger} )Y^{-1/2}$ 
and $\Gamma =Y^{1/2}\re(U_{1}U_{2}^{\dagger}\Lambda U_{2}U_{1}^{\dagger} )Y^{1/2}$ in~\eqref{G},~\eqref{C}. 
Then it can be  verified that $R=R^{\top}$, $\Gamma=-\Gamma^{\top}$. 
Moreover, substituting $Q=-iRY+Y^{-1}\Gamma = Y^{-1/2}U_{1}U_{2}^{\dagger}\Lambda U_{2}U_{1}^{\dagger} Y^{1/2}$ into~\eqref{rankconstraint} yields 
 \begin{align*}
 &\rank\left(\begin{bmatrix}P &QP &\cdots &Q^{N-1}P\end{bmatrix}\right)\\
=&\rank \left( \left[
U_{2}U_{1}^{\dagger}Y^{1/2}\Upsilon^{\top} \;\Lambda U_{2}U_{1}^{\dagger}Y^{1/2}\Upsilon^{\top} \cdots \right.\right. \\
&\hspace{4cm}\left.\left.\;\Lambda^{N-1} U_{2}U_{1}^{\dagger}Y^{1/2}\Upsilon^{\top}
\right]\right)\\
=&\rank \left(\frac{\norm{Y^{1/2}\Upsilon^{\top}}}{\sqrt{N}} \begin{bmatrix}1 &(i\alpha_{1}) &\cdots &(i\alpha_{1})^{N-1}\\1 &(i\alpha_{2}) &\cdots &(i\alpha_{2})^{N-1}\\
\vdots &\vdots &\cdots &\vdots\\
1 &(i\alpha_{N}) &\cdots &(i\alpha_{N})^{N-1}\\\end{bmatrix}\right)\\
=&N.
 \end{align*}
Here we have used the full rank property of a Vandermonde matrix. 
Hence the rank condition~\eqref{rankconstraint} is satisfied. Based on Lemma~\ref{lem1}, we now  obtain an $N$-mode locally dissipative quantum system that is asymptotically stable and achieves the given covariance matrix.  The coupling vector   $\hat{L}$ of the  system, which consists of only one Lindblad operator, is given by 
\begin{align*}
C&=P^{\top}\left[-Z \;\; I_{N}\right]=\left[-\Upsilon Z \;\; \Upsilon\right]\\
&=\begin{bmatrix}
0_{1\times (\ell-1)} &-\tau_{2} & 0_{1\times (N-\ell)} & 0_{1\times (\ell-1)} &\tau_{1} & 0_{1\times (N-\ell)}
\end{bmatrix},\\
\hat{L}&=C\hat{x}=-\tau_{2}\hat{q}_{\ell}+\tau_{1}\hat{p}_{\ell}.
\end{align*}
We see that $\hat{L}$ acts only on the $\ell$th mode, and hence it is local. The  Hamiltonian $\hat{H}$ of the system can also be obtained by directly substituting the matrices $R$ and $\Gamma$ above into~\eqref{G}.  
This completes the sufficiency part of the proof.

Next we prove the necessity part. 
Suppose an $N$-mode pure Gaussian state can be uniquely generated in an $N$-mode locally dissipative quantum system. 
Based on Lemma~\ref{lem1}, there exists a $P\in \mathbb{C}^{N\times K}$, 
$P\ne 0$, such that the coupling vector   $
\hat{L}=P^{\top}\left[-Z\;\; I_{N}\right]\hat{x} $
is local. 
Let $P_{k}$, $1\le k\le K$, be a nonzero column in the matrix $P$. 
Then the corresponding Lindblad operator 
$\hat{c}_{k}=P_{k}^{\top}\left[-Z\; I_{N}\right]\hat{x}$ is local. Suppose that $\hat{c}_{k}$ acts on the $\ell$th mode of the system. Then we have 
\begin{align*}
&P_{k}^{\top}\left[-Z\;\; I_{N}\right]\\
=&\begin{bmatrix}
0_{1\times (\ell-1)} &\tau_{3} & 0_{1\times (N-\ell)} & 0_{1\times (\ell-1)} &\tau_{4} & 0_{1\times (N-\ell)}
\end{bmatrix},
\end{align*}
where $\tau_{3}$ and $\tau_{4}$ are complex numbers. 
It follows that 
  \begin{align}
   P_{k}^{\top}Z&= \begin{bmatrix}
 0_{1\times (\ell-1)} &-\tau_{3}  & 0_{1\times (N-\ell)}
 \end{bmatrix},\label{pf1}\\
 P_{k}^{\top}&= \begin{bmatrix}
 0_{1\times (\ell-1)} &\tau_{4} & 0_{1\times (N-\ell)} 
 \end{bmatrix},\quad \tau_{4}\ne 0.\label{pf2}
  \end{align}
Substituting~\eqref{pf2} into~\eqref{pf1} gives
  \begin{align*}
    P_{k}^{\top}Z
   =&\tau_{4}\begin{bmatrix}
  Z_{(\ell, 1)}  &Z_{(\ell, 2)} &\cdots &Z_{(\ell, N)}
  \end{bmatrix}\\
  =&\begin{bmatrix}
 0_{1\times (\ell-1)} &-\tau_{3}  & 0_{1\times (N-\ell)}
 \end{bmatrix}.
  \end{align*} 
Since $\tau_{4}\ne 0$, we have $Z_{(\ell, j)}=0$, $\forall j\ne \ell$. Since $Z=Z^{\top}$, we have $Z_{(\ell, j)}=Z_{(j,\ell)}=0$, $\forall j\ne \ell$. 
That is, Equation~\eqref{thm21} holds. 
This completes the proof. \hfill  
\end{proof}

\begin{rem}
The basic idea of Theorem~\ref{thm2} is that for any choice of $P\ne 0$, 
there always exist matrices $R=R^{\top}$ and $\Gamma=-\Gamma^{\top}$ 
such that the rank condition~\eqref{rankconstraint} is satisfied. 
So we can first specify a matrix $P$ such that the coupling matrix $C$ in~\eqref{C}
has a local structure. 
After obtaining $P$, we determine the other two matrices $R$ and 
$\Gamma$ to get a system Hamiltonian, under the rank constraint  
\eqref{rankconstraint}. 
Generally, for a given nonzero matrix $P$, we have infinite solutions 
$\left(R,\;\;\Gamma\right)$ that satisfy the rank condition 
\eqref{rankconstraint}. 
Different choices of $\left(R,\;\;\Gamma\right)$ lead to different system 
Hamiltonians. 
The optimization problem over these Hamiltonians is beyond 
the scope of this paper and is not considered, but in the next subsection we 
will show a specific recipe for determining those matrices $\left(R,\;\;\Gamma\right)$. 
\end{rem}

\begin{rem}
Suppose an $N$-mode pure Gaussian state is generated in an $N$-mode dissipative quantum system and the $\ell$th mode is locally coupled to the environment. Then from Equation~\eqref{thm21}, it is straightforward to see that the $\ell$th mode is not entangled with the rest of the system modes when the system achieves the steady state. 
\end{rem}
  
\subsection{Examples} 
\begin{exmp} \label{exam2}
We consider the generation  of {\it canonical Gaussian cluster states}, 
which serve as an essential resource in quantum computation with continuous variables~\cite{BP03:book,MFL11:pra,MLGWRN06:prl}. We mention that an interesting class of cluster states, called \emph{bilayer square-lattice continuous--variable cluster states}, has been proposed recently in~\cite{AWSCPM15:arxiv}. This class of cluster states has some practical advantages over canonical Gaussian cluster states for quantum computation~\cite{AWSCPM15:arxiv}. For the sake of simplicity, we use canonical Gaussian cluster states to illustrate the developed theory. 
The covariance matrix  $V$ corresponding to an $N$-mode canonical Gaussian cluster state is given by
$
V=\frac{1}{2}\begin{bmatrix}
e^{2\alpha}I_{N} &e^{2\alpha}B\\
e^{2\alpha}B &e^{-2\alpha}I_{N}+e^{2\alpha}B^{2}
\end{bmatrix}$, 
where $B=B^{\top}\in\mathbb{R}^{ N\times N}$ and $\alpha $ is 
the squeezing parameter. 
Note that in the limit $\alpha\rightarrow \infty$, the canonical Gaussian cluster state  
approximates the corresponding ideal cluster state. 
Using~\eqref{covariance}, we obtain $X=B$ and $Y=e^{-2\alpha}I_{N}$. The graph corresponding to a canonical Gaussian cluster state is given by $Z=X+iY=B+ie^{-2\alpha}I_{N}$.

Let us consider a simple case where 
\begin{align}
\label{first cluster example}
X=\begin{bmatrix}
0 &1 &0 &0\\
1 &0 &1  &0\\
0 &1 &0 &0\\
0 &0 &0 &\sqrt{2}
\end{bmatrix}, \quad Y=e^{-2\alpha}I_{4}. 
\end{align}
These matrices satisfy $X_{4j}=0$ and $Y_{4j}=0$ for all $j\neq 4$. Thus by
Theorem~\ref{thm2}, the corresponding canonical Gaussian cluster state can be generated
in a $four$-mode locally dissipative system.  
To construct such a  system, let us take $P=[0\; 0\; 0\; 1]^\top$ in Lemma~\ref{lem1}. 
The next step is to determine the other system parameters 
$R$ and $\Gamma$. 
For a practical implementation, one of the basic requirements on the system is that, as mentioned before, the system has as few  pumped  crystals as possible. 
Motivated by the structure of the passive quantum systems described in Remark~\ref{rmk2}, 
we choose $R=0_{4\times 4}$. 
As a result, the Hamiltonian matrix is $
        M=\begin{bmatrix}
                -e^{2\alpha}(\Gamma X+X\Gamma^\top) & e^{2\alpha}\Gamma \\
                 e^{2\alpha}\Gamma^\top & 0_{4\times 4} \\
              \end{bmatrix}$. 
The $(1,2)$ block in $M$ is a skew matrix 
$e^{2\alpha}\Gamma$. So if we can additionally take the $(1,1)$ block to be a diagonal matrix, 
then the interaction Hamiltonian  
between the modes is passive and can be simply realized by beam splitters.
According to this guideline, we now seek $\Gamma$ such that $-e^{2\alpha}(\Gamma X+X\Gamma^\top)$ is a diagonal matrix. By direct calculation, we obtain
\begin{align*}
       \Gamma
            =\begin{bmatrix}
                 0 & \gamma_{1} & 0 & \gamma_{2} \\
                 -\gamma_{1} & 0 & \gamma_{1} & \sqrt{2}\gamma_{2}\\
                 0 & -\gamma_{1} & 0 & \gamma_{2} \\
                 -\gamma_{2} & -\sqrt{2}\gamma_{2} & -\gamma_{2} & 0 \\
              \end{bmatrix},
\end{align*}              
where $\gamma_{1}\in \mathbb{R}$ and $\gamma_{2}\in \mathbb{R}$. Substituting the  matrices $P$, $R$ and $\Gamma$ above into the rank condition~\eqref{rankconstraint}, we obtain that 
if $\gamma_{1}\gamma_{2}\neq 0$, the resulting linear quantum system is asymptotically stable 
and achieves the covariance matrix corresponding to~\eqref{first cluster example}. 

The Hamiltonian of this linear quantum system is now determined as 
\begin{align*}
     \hat{H}&= -\gamma_{1}e^{2\alpha}\hat{q}_1^2
            + \gamma_{1}e^{2\alpha}\hat{q}_3^2
             + e^{2\alpha}\gamma_{1}
                    (\hat{H}_{12}^{\rm{(BS)}} + \hat{H}_{23}^{\rm{(BS)}})\\
             &\quad + e^{2\alpha}\gamma_{2}
                    (\hat{H}_{14}^{\rm{(BS)}} + \sqrt{2}\hat{H}_{24}^{\rm{(BS)}}
                       + \hat{H}_{34}^{\rm{(BS)}}),
\end{align*}
where $
     \hat{H}_{jk}^{\rm{(BS)}} = (\hat{q}_j\hat{p}_k - \hat{p}_j\hat{q}_k)
       = i(\hat{a}_j\hat{a}_k^\ast - \hat{a}_j^\ast\hat{a}_k)$, where $\hat{a}_j=(\hat{q}_j+i\hat{p}_j)/\sqrt{2}$ and $\hat{a}_j^\ast=(\hat{q}_j-i\hat{p}_j)/\sqrt{2}$, 
is the Hamiltonian representing the coupling between the $j$th and $k$th 
optical modes at a beam splitter. 
Also the coupling vector   is given by 
\begin{align*}
    \hat{L}= -(\sqrt{2} +e^{-2\alpha}i) \hat{q}_4 + \hat{p}_4,
\end{align*}
which acts only on the fourth mode and hence it is local. 
Finally, using the result in~\cite{NJD09:siamjco}, a corresponding optical realization of this linear quantum system is shown in Fig.~\ref{canonicalstate1}. 
Note that three pumped  crystals are used; we conjecture 
that this is the minimum number required for constructing a desired locally dissipative 
system. 
\begin{figure}[htbp]
\begin{center}
\includegraphics[height=5.292cm]{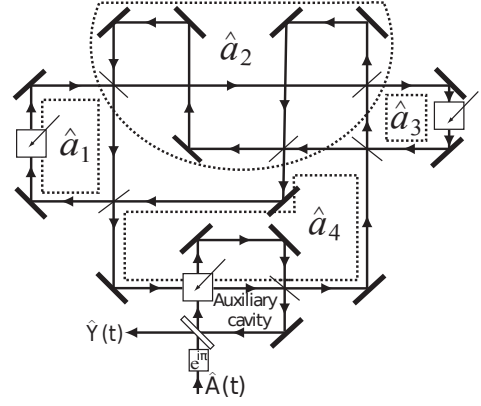}
\caption{
The optical dissipative system that uniquely generates 
the canonical Gaussian cluster state \eqref{first cluster example}. 
The coupling vector  $\hat{L}$ acts only on the fourth mode, and hence it is local. 
}
\label{canonicalstate1}
\end{center}
\end{figure}
\end{exmp}

\begin{exmp}
We next consider a canonical Gaussian cluster state specified by the following 
matrices $X$ and $Y$: 
\begin{align}
\label{second cluster example}
X=\begin{bmatrix}
0 &1 &0 &0\\
1 &0 &1  &0\\
0 &1 &0 &1\\
0 &0 &1 &0
\end{bmatrix}, \quad Y=e^{-2\alpha}I_{4}. 
\end{align} 
The strength of Theorem~\ref{thm2} is  that it readily tells us that 
this canonical Gaussian cluster state cannot be generated in {\it any} four-mode  locally dissipative 
system. 
Nonetheless let us take the same matrix $P$  as before, i.e., $P=[0\; 0\; 0\; 1]^\top$, 
and follow the same guideline as discussed in Example~\ref{exam2}.  
That is, we set $R=0_{4\times 4}$ and seek $\Gamma$ such that $\Gamma X+ X\Gamma^\top$ 
is a diagonal matrix. 
Then, again by direct calculation, 
we find   
\begin{align*}
       \Gamma
            &=                 \begin{bmatrix}
                 0 & \gamma_{1} & 0 & \gamma_{2} \\
                 -\gamma_{1} & 0 & \gamma_{1}+\gamma_{2} & 0 \\
                 0 & -\left(\gamma_{1}+\gamma_{2}\right) & 0 & \gamma_{1} \\
                - \gamma_{2} & 0 & -\gamma_{1} & 0 \\
              \end{bmatrix}.
\end{align*}

Let us take  $\gamma_{1}=1$  and $\gamma_{2}=0$. Then the corresponding system Hamiltonian is given by 
\begin{equation}
\label{Example 2; Hamiltonian}
     \hat{H} 
         = -e^{2\alpha}(\hat{q}_1^2 - \hat{q}_4^2)
             + e^{2\alpha}(\hat{H}_{12}^{\rm{(BS)}} + \hat{H}_{23}^{\rm{(BS)}} 
                 + \hat{H}_{34}^{\rm{(BS)}}).
\end{equation}
It can be verified that the rank condition \eqref{rankconstraint} is satisfied, 
hence the system constructed here is asymptotically stable 
and achieves the desired covariance matrix corresponding to~\eqref{second cluster example}, though in this case 
the system needs to have the following non-local interaction with its environment: 
\begin{align*}
    \hat{L}&= \begin{bmatrix}
          0 & 0 & -1 & -e^{-2\alpha}i & \vline  &0   &0   &0   &1
    \end{bmatrix}\hat{x}= -\hat{q}_3 -i e^{-2\alpha} \hat{q}_4 + \hat{p}_4. 
\end{align*}
An optical realization, which yet contains an abstract component corresponding 
to this non-local interaction, is depicted in Fig.~\ref{canonicalstate2}. 
\begin{figure}[htbp]
\begin{center}
\includegraphics[width=6.6cm]{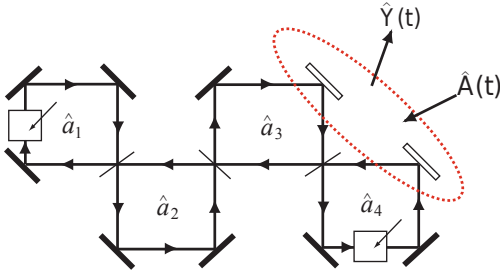}
\caption{
The optical linear quantum system that uniquely 
generates the canonical Gaussian cluster state \eqref{second cluster example}. 
The coupling vector  $\hat{L}$ acts on the third and fourth modes, and hence 
it is not local. }
\label{canonicalstate2}
\end{center}
\end{figure}
A practical implementation of  the non-local interaction depicted in 
Fig.~\ref{canonicalstate2} could be experimentally difficult. 
Nonetheless this issue can be resolved by taking the following method: 
add an auxiliary system with a single mode $\hat x_A=\left[\hat{q}_{A}\;\hat{p}_{A}\right]^{\top},$ 
and specify the target canonical Gaussian cluster state as 
\begin{align}
\label{second cluster example enlarged}
\tilde{X}=\begin{bmatrix}
0 &1 &0 &0 &0 \\
1 &0 &1  &0 &0\\
0 &1 &0 &1 &0\\
0 &0 &1 &0 &0\\
0 &0 &0 & 0 & \lambda \\
\end{bmatrix}, \quad \tilde{Y}=e^{-2\alpha}I_{5}. 
\end{align} 
Since $\tilde{X}_{5j} =0$ and $\tilde{Y}_{5j} = 0$ for all $ j\neq 5$, 
by Theorem~\ref{thm2}, we can construct a five-mode locally dissipative system that uniquely
generates the above canonical Gaussian cluster state~\eqref{second cluster example enlarged}.  
By choosing $P=[0\;0\;0\;0\;1]^\top$ and then taking a similar procedure 
as in the case of Example~\ref{exam2}, we can obtain such a desired locally dissipative quantum system. 
Now we obtain an important observation:  
for an $N$-mode canonical Gaussian cluster state with the graph matrix 
$X=B$ and the squeezing matrix $Y=e^{-2\alpha}I_N$, it is {\it always} 
possible to generate this state in a locally dissipative quantum system by  adding 
a single-mode auxiliary system and specifying the target state as 
$\tilde{X}={\rm diag}[B, \lambda]$ and $\tilde{Y}=e^{-2\alpha}I_{N+1}$. 

\begin{rem}
The method in~\cite{IY13:pra} is based on essentially the same idea; given an $N$-mode pure Gaussian state with graph $Z=X+iY$, instead of generating it directly, 
we enlarge the system by adding a single-mode auxiliary system 
and then specify the target state as $\tilde{X}={\rm diag}[X, \lambda]$ 
and $\tilde{Y}={\rm diag}[Y, 1]$.  
By Theorem~\ref{thm2}, this $(N+1)$-mode target state can be uniquely generated 
in an $(N+1)$-mode locally dissipative system. 
The original $N$-mode pure Gaussian state is then obtained as a reduced state of the target state. 
\end{rem}
\end{exmp}


\section{Conclusion} \label{Conclusion}
In this paper, we have provided two feasible realizations of linear quantum systems for covariance assignment corresponding to
 pure Gaussian states: 
a cascade realization and a locally dissipative realization. 
First, we have shown that given any covariance matrix corresponding to a pure Gaussian state, we can construct a cascaded quantum system that achieves the assigned covariance matrix. This cascaded quantum system is constructed as a cascade connection of several one-dimensional harmonic oscillators,
without any direct interaction Hamiltonians between these oscillators.
Second, we have given a complete characterization of the class of pure Gaussian 
 states that can be generated using  locally dissipative quantum systems. 
In particular, we have shown a specific recipe for constructing a system 
having a relatively simple Hamiltonian coupling between the system modes. 
The results developed in this paper are potentially useful for
the preparation of pure Gaussian states. In the examples, we have provided  realizations of $(\hat{H},\;\hat{L})$ in quantum optics using the result in~\cite{NJD09:siamjco}. The circuit figures shown in the examples are not necessarily the simplest realizations in quantum optics.  Also, a system with $(\hat{H},\;\hat{L})$ could be realized by  other instances of linear quantum systems such as atomic ensembles and optomechanical systems~\cite{MPC11:pra,WC14:pra}.


\section*{Appendix}

Here we briefly review the synthesis theory of linear quantum systems in quantum optics  developed in~\cite{NJD09:siamjco}. 
\subsection*{1. Realization of a quadratic Hamiltonian }
Suppose a quadratic Hamiltonian is given by $\hat{H}_{d}=\frac{1}{2}\hat{\xi}^{\top} M_{d} \hat{\xi}$, where $\hat{\xi}=[
\hat{q}\;\;\hat{p}
]^{\top}$ and $M_{d}=M_{d}^{\top}\in \mathbb{R}^{2\times 2}$.  This Hamiltonian   can be realized by placing a crystal with a classical pump inside an optical cavity, as shown in Fig.~\ref{appendix1}. Working in the frame rotating at half the pump frequency, the Hamiltonian is written as 
\begin{align}
\hat{H}_{r}&=\triangle\hat{a}^{\ast} \hat{a} +\frac{i}{2} \left(\epsilon(\hat{a}^{\ast})^{ 2}- \epsilon^{\ast}\hat{a}^{2}  \right) \notag \\
&= \frac{1}{2}\hat{\xi}^{\top}\begin{bmatrix}
\triangle-\im(\epsilon) &\re(\epsilon)\\
\re(\epsilon) &\triangle+\im(\epsilon)
\end{bmatrix} \hat{\xi} -  \frac{\triangle}{2}, \label{Hrealization}
\end{align}
where $\triangle=\omega_{\text{cav}}-\omega_{p}/2$ is the detuning between the cavity mode frequency  and the half pump frequency.  $\epsilon$ is a measure of the effective pump intensity~\cite{GZ00:book}. From~\eqref{Hrealization}, we see  that by choosing the values of $\triangle$ and $\epsilon$, one can make $\hat{H}_{r}=\hat{H}_{d}-  \frac{\triangle}{2}$. Note that the constant term $-  \frac{\triangle}{2}$ does not affect the dynamics of a linear quantum system, and hence can be ignored. Therefore,  the desired Hamiltonian $\hat{H}_{d}$ can be  realized in this scheme. 
\begin{figure}[htbp]
\begin{center}
\includegraphics[width=5cm]{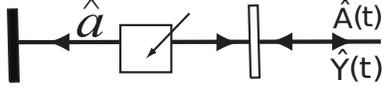}
\caption{A quadratic Hamiltonian   can be realized by placing a crystal with a classical pump inside an optical cavity. }
\label{appendix1}
\end{center}
\end{figure}

\subsection*{2. Realization of a beam-splitter-like interaction Hamiltonian}
Suppose a Hamiltonian is given by $\hat{H}_{d}=h_{d}\hat{a}_{1}^{\ast}\hat{a}_{2}+h_{d}^{\ast}\hat{a}_{2}^{\ast}\hat{a}_{1}$, where $h_{d}\in \mathbb{C}$. This Hamiltonian can be realized by implementing a beam splitter for the two incoming modes $\hat{a}_{1}$ and $\hat{a}_{2}$, as shown in Fig.~\ref{appendix2}. At the beam splitter, we have the following transformations
\begin{align*}
\begin{bmatrix}
\hat{a}_{3}\\
\hat{a}_{4}
\end{bmatrix}=\begin{bmatrix}
t_{2} &r_{1}\\
r_{2} &t_{1}
\end{bmatrix}\begin{bmatrix}
\hat{a}_{1}\\
\hat{a}_{2}
\end{bmatrix}, 
\end{align*}
where $\hat{a}_{3}$ and $\hat{a}_{4}$ denote the outgoing modes, and $r_{1},\; t_{1} \in \mathbb{C}$ denote the (complex) reflectance and transmittance of the beam splitter,  respectively. Note that  $r_{1}$, $t_{1}$, $r_{2}$ and $t_{2}$   satisfy the following relations:
$|r_{2}|=|r_{1}|$, $|t_{2}|=|t_{1}|$,  $|r_{1}|^{2}+|t_{1}|^{2}=1$, $r_{1}^{\ast}t_{2}+r_{2}t_{1}^{\ast}=0$, and  $r_{1}^{\ast}t_{1}+r_{2}t_{2}^{\ast}=0$~\cite{GK04:book}.
Let us parametrize them as $
r_{1}=e^{-i\phi} sin \theta$, $r_{2}=- e^{i\phi} sin \theta$, and $t_{1}=t_{2}=cos\theta$.
Then the interaction Hamiltonian $\hat{H}^{\rm{(BS)}}$  for
this beam splitter is given by
\begin{align}
\hat{H}^{\rm{(BS)}} = i\theta e^{-i\phi} \hat{a}_{1}^{\ast} \hat{a}_{2} -i \theta e^{i\phi} \hat{a}_{2}^{\ast}\hat{a}_{1}. \label{Hrealization2} 
\end{align}
From~\eqref{Hrealization2}, we see  that by choosing the values of $\theta $ and $\phi$, one can make $\hat{H}^{\rm{(BS)}}=\hat{H}_{d}$. That is, the desired beam-splitter-like interaction Hamiltonian $\hat{H}_{d}$ can be realized in this scheme. 
\begin{figure}[htbp]
\begin{center}
\includegraphics[width=4cm]{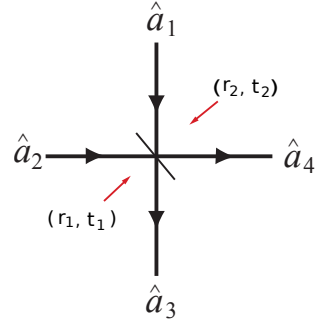}
\caption{A beam-splitter-like interaction Hamiltonian   can be realized by placing a beam splitter for the two incoming modes $\hat{a}_{1}$ and $\hat{a}_{2}$. }
\label{appendix2}
\end{center}
\end{figure}

\subsection*{3. Realization of a dissipative coupling $\hat{L}$}
To realize a coupling operator $\hat{L}_{d}=c_{1}\hat{q}+c_{2}\hat{p}=\left(\frac{c_{1}-ic_{2}}{\sqrt{2}}\right)\hat{a}+\left(\frac{c_{1}+ic_{2}}{\sqrt{2}}\right)\hat{a}^{\ast}$, we consider the configuration shown in Fig.~\ref{appendix3}. The configuration consists of a ring cavity with mode $\hat{a}$ and an auxiliary ring cavity with mode $\hat{b}$.  The
cavity modes $\hat{a}$ and $\hat{b}$ interact through a crystal pumped by  a classical  beam,
and a beam splitter. The frequency of the auxiliary cavity mode $\hat{b}$ is matched to half the  pump frequency. Working in the frame rotating at half the pump frequency, the interaction Hamiltonian is written as
\begin{align*}
\hat{H}_{ab}=\frac{i}{2}\left(\epsilon_{1} \hat{a}^{\ast}\hat{b}^{\ast}- \epsilon_{1}^{\ast} \hat{a}\hat{b} \right)+\frac{i}{2}\left(\epsilon_{2} \hat{a}^{\ast}\hat{b}- \epsilon_{2}^{\ast} \hat{a} \hat{b}^{\ast}\right),
\end{align*}
where $\epsilon_{1}$ determines the effective pump intensity and $\epsilon_{2}$ determines the parameters of the beam splitter. Assume that the coupling coefficient $\gamma$ of the partially transmitting mirror is large so that  the mode $\hat{b}$ is heavily damped and  can be adiabatically eliminated. Then after elimination of $\hat{b}$, the resulting coupling operator is given by 
\begin{align}
\hat{L}_{r}=\frac{1}{\sqrt{\gamma}}(-\epsilon_{2}^{\ast}\hat{a}+\epsilon_{1}\hat{a}^{\ast}). \label{Lrealization} 
\end{align}
From~\eqref{Lrealization}, we see  that by choosing the values of $\epsilon_{1}$, $\epsilon_{2}$, and $\gamma $ with $\gamma $ being large,  we can make $\hat{L}_{r}=\hat{L}_{d}$. That is,  the desired coupling operator $\hat{L}_{d}$ can be realized in this scheme. See~\cite{NJD09:siamjco} for details.   
\begin{figure}[htbp]
\begin{center}
\includegraphics[width=5cm]{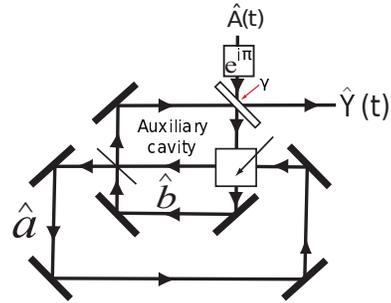}
\caption{Realization of a dissipative coupling operator $\hat{L}$.  }
\label{appendix3}
\end{center}
\end{figure}



\begin{thebibliography}{10}
\providecommand{\url}[1]{#1}
\csname url@samestyle\endcsname
\providecommand{\newblock}{\relax}
\providecommand{\bibinfo}[2]{#2}
\providecommand{\BIBentrySTDinterwordspacing}{\spaceskip=0pt\relax}
\providecommand{\BIBentryALTinterwordstretchfactor}{4}
\providecommand{\BIBentryALTinterwordspacing}{\spaceskip=\fontdimen2\font plus
\BIBentryALTinterwordstretchfactor\fontdimen3\font minus
  \fontdimen4\font\relax}
\providecommand{\BIBforeignlanguage}[2]{{%
\expandafter\ifx\csname l@#1\endcsname\relax
\typeout{** WARNING: IEEEtran.bst: No hyphenation pattern has been}%
\typeout{** loaded for the language `#1'. Using the pattern for}%
\typeout{** the default language instead.}%
\else
\language=\csname l@#1\endcsname
\fi
#2}}
\providecommand{\BIBdecl}{\relax}
\BIBdecl

\bibitem{HS87:ijc}
A.~Hotz and R.~E. Skelton, ``Covariance control theory,'' \emph{International
  Journal of Control}, vol.~46, no.~1, pp. 13--32, 1987.

\bibitem{CS87:tac}
J.~E.~G.~Collins and R.~E. Skelton, ``A theory of state covariance assignment
  for discrete systems,'' \emph{IEEE Transactions on Automatic Control},
  vol.~32, no.~1, pp. 35--41, 1987.

\bibitem{SI89:ijc}
R.~E. Skelton and M.~Ikeda, ``Covariance controllers for linear continuous-time
  systems,'' \emph{International Journal of Control}, vol.~49, no.~5, pp.
  1773--1785, 1989.

\bibitem{BP03:book}
S.~L. Braunstein and A.~K. Pati, \emph{Quantum Information with Continuous
  Variables}.\hskip 1em plus 0.5em minus 0.4em\relax Springer, 2003.

\bibitem{weedbrook12:rmp}
C.~Weedbrook, S.~Pirandola, R.~Garc\'ia-Patr\'on, N.~J. Cerf, T.~C. Ralph,
  J.~H. Shapiro, and S.~Lloyd, ``{G}aussian quantum information,''
  \emph{Reviews of Modern Physics}, vol.~84, no.~2, pp. 621--669, 2012.

\bibitem{OHY11:cdc}
K.~Ohki, S.~Hara, and N.~Yamamoto, ``On quantum-classical equivalence for
  linear systems control problems and its application to quantum entanglement
  assignment,'' in \emph{Proceedings of IEEE 50th Annual Conference on Decision
  and Control (CDC)}, December 2011, pp. 6260--6265.

\bibitem{KY12:pra}
K.~Koga and N.~Yamamoto, ``Dissipation-induced pure {G}aussian state,''
  \emph{Physical Review A}, vol.~85, no.~2, p. 022103, 2012.

\bibitem{Y12:ptrsa}
N.~Yamamoto, ``Pure {G}aussian state generation via dissipation: a quantum
  stochastic differential equation approach,'' \emph{Philosophical Transactions
  of the Royal Society A: Mathematical, Physical and Engineering Sciences},
  vol. 370, no. 1979, pp. 5324--5337, 2012.

\bibitem{IY13:pra}
Y.~Ikeda and N.~Yamamoto, ``Deterministic generation of {G}aussian pure states
  in a quasilocal dissipative system,'' \emph{Physical Review A}, vol.~87,
  no.~3, p. 033802, 2013.

\bibitem{MWPY14:msc}
S.~Ma, M.~J. Woolley, I.~R. Petersen, and N.~Yamamoto, ``Preparation of pure
  {G}aussian states via cascaded quantum systems,'' in \emph{Proceedings of
  IEEE Conference on Control Applications (CCA)}, October 2014, pp. 1970--1975.

\bibitem{MFL11:pra}
N.~C. Menicucci, S.~T. Flammia, and P.~van Loock, ``Graphical calculus for
  {G}aussian pure states,'' \emph{Physical Review A}, vol.~83, no.~4, p.
  042335, 2011.

\bibitem{MLGWRN06:prl}
N.~C. Menicucci, P.~van Loock, M.~Gu, C.~Weedbrook, T.~C. Ralph, and M.~A.
  Nielsen, ``Universal quantum computation with continuous-variable cluster
  states,'' \emph{Physical Review Letters}, vol.~97, no.~11, p. 110501, 2006.

\bibitem{A06:prl}
G.~Adesso, ``Generic entanglement and standard form for {N}-mode pure
  {G}aussian states,'' \emph{Physical Review Letters}, vol.~97, p. 130502,
  2006.

\bibitem{CPBZ93:prl}
J.~I. Cirac, A.~S. Parkins, R.~Blatt, and P.~Zoller, ```{D}ark' squeezed states
  of the motion of a trapped ion,'' \emph{Physical Review Letters}, vol.~70,
  no.~5, pp. 556--559, 1993.

\bibitem{PCZ96:prl}
J.~F. Poyatos, J.~I. Cirac, and P.~Zoller, ``Quantum reservoir engineering with
  laser cooled trapped ions,'' \emph{Physical Review Letters}, vol.~77, no.~23,
  pp. 4728--4731, 1996.

\bibitem{WC13:prl}
Y.~D. Wang and A.~A. Clerk, ``Reservoir-engineered entanglement in
  optomechanical systems,'' \emph{Physical Review Letters}, vol. 110, no.~25,
  p. 253601, 2013.

\bibitem{KM11:prl}
H.~Krauter, C.~A. Muschik, K.~Jensen, W.~Wasilewski, J.~M. Petersen, J.~I.
  Cirac, and E.~S. Polzik, ``Entanglement generated by dissipation and steady
  state entanglement of two macroscopic objects,'' \emph{Physical Review
  Letters}, vol. 107, no.~8, p. 080503, 2011.

\bibitem{WC14:pra}
M.~J. Woolley and A.~A. Clerk, ``Two-mode squeezed states in cavity
  optomechanics via engineering of a single reservoir,'' \emph{Physical Review
  A}, vol.~89, no.~6, p. 063805, 2014.

\bibitem{Gardiner93:prl}
C.~W. Gardiner, ``Driving a quantum system with the output field from another
  driven quantum system,'' \emph{Physical Review Letters}, vol.~70, no.~15, pp.
  2269--2272, 1993.

\bibitem{N10:tac}
H.~I. Nurdin, ``On synthesis of linear quantum stochastic systems by pure
  cascading,'' \emph{IEEE Transactions on Automatic Control}, vol.~55, no.~10,
  pp. 2439--2444, 2010.

\bibitem{P11:auto}
I.~R. Petersen, ``Cascade cavity realization for a class of complex transfer
  functions arising in coherent quantum feedback control,'' \emph{Automatica},
  vol.~47, no.~8, pp. 1757--1763, 2011.

\bibitem{SS90:scl}
P.~Seibert and R.~Suarez, ``Global stabilization of nonlinear cascade
  systems,'' \emph{Systems $\&$ Control Letters}, vol.~14, no.~4, pp. 347--352,
  1990.

\bibitem{HJJ05:scl}
S.~Huang, M.~R. James, and Z.~P. Jiang, ``${L}_\infty$-bounded robust control
  of nonlinear cascade systems,'' \emph{Systems $\&$ Control Letters}, vol.~54,
  no.~3, pp. 215--224, 2005.

\bibitem{LH06:tac}
L.~Liu and J.~Huang, ``Global robust stabilization of cascade-connected systems
  with dynamic uncertainties without knowing the control direction,''
  \emph{IEEE Transactions on Automatic Control}, vol.~51, no.~10, pp.
  1693--1699, 2006.

\bibitem{KBDKMZ08:pra}
B.~Kraus, H.~P. B\"uchler, S.~Diehl, A.~Kantian, A.~Micheli, and P.~Zoller,
  ``Preparation of entangled states by quantum {M}arkov processes,''
  \emph{Physical Review A}, vol.~78, no.~4, p. 042307, 2008.

\bibitem{RLMM12:pra}
M.~Rafiee, C.~Lupo, H.~Mokhtari, and S.~Mancini, ``Stationary and uniform
  entanglement distribution in qubit networks with quasilocal dissipation,''
  \emph{Physical Review A}, vol.~85, p. 042320, 2012.

\bibitem{TV12:ptrsa}
F.~Ticozzi and L.~Viola, ``Stabilizing entangled states with quasi-local
  quantum dynamical semigroups,'' \emph{Philosophical Transactions of the Royal
  Society A: Mathematical, Physical and Engineering Sciences}, vol. 370, no.
  1979, pp. 5259--5269, 2012.

\bibitem{BR04:book}
H.~A. Bachor and T.~C. Ralph, \emph{A Guide to Experiments in Quantum
  Optics}.\hskip 1em plus 0.5em minus 0.4em\relax Wiley, 2004.

\bibitem{TFSBK14:prl}
T.~Tufarelli, A.~Ferraro, A.~Serafini, S.~Bose, and M.~S. Kim, ``Coherently
  opening a high-${Q}$ cavity,'' \emph{Physical Review Letters}, vol. 112,
  no.~13, p. 133605, 2014.

\bibitem{TA14:jpa}
G.~T\'oth and I.~Apellaniz, ``Quantum metrology from a quantum information
  science perspective,'' \emph{Journal of Physics A: Mathematical and
  Theoretical}, vol.~47, no.~42, p. 424006, 2014.

\bibitem{WM10:book}
H.~M. Wiseman and G.~J. Milburn, \emph{Quantum Measurement and Control}.\hskip
  1em plus 0.5em minus 0.4em\relax Cambridge University Press, 2010.

\bibitem{HP84:cmp}
R.~L. Hudson and K.~R. Parthasarathy, ``Quantum {I}to's formula and stochastic
  evolutions,'' \emph{Communications in Mathematical Physics}, vol.~93, no.~3,
  pp. 301--323, 1984.

\bibitem{B92:JMA}
V.~P. Belavkin, ``Quantum stochastic calculus and quantum nonlinear
  filtering,'' \emph{Journal of Multivariate Analysis}, vol.~42, no.~2, pp.
  171--201, 1992.

\bibitem{GZ00:book}
C.~W. Gardiner and P.~Zoller, \emph{Quantum Noise: A Handbook of {M}arkovian
  and Non-{M}arkovian Quantum Stochastic Methods with Applications to Quantum
  Optics}.\hskip 1em plus 0.5em minus 0.4em\relax Springer, 2000.

\bibitem{BHJ07:siamjco}
L.~Bouten, R.~V. Handel, and M.~R. James, ``An introduction to quantum
  filtering,'' \emph{SIAM Journal on Control and Optimization}, vol.~46, no.~6,
  pp. 2199--2241, 2007.

\bibitem{WGKWC04:pra}
M.~M. Wolf, G.~Giedke, O.~Kr{\"u}ger, R.~F. Werner, and J.~I. Cirac,
  ``{G}aussian entanglement of formation,'' \emph{Physical Review A}, vol.~69,
  no.~5, p. 052320, 2004.

\bibitem{TN15:scl}
O.~Techakesari and H.~I. Nurdin, ``On the quasi-balanceable class of linear
  quantum stochastic systems,'' \emph{Systems $\&$ Control Letters}, vol.~78,
  pp. 25--31, 2015.

\bibitem{H14:tac}
H.~I. Nurdin, ``Structures and transformations for model reduction of linear
  quantum stochastic systems,'' \emph{IEEE Transactions on Automatic Control},
  vol.~59, no.~9, pp. 2413--2425, 2014.

\bibitem{P12:scl}
I.~R. Petersen, ``Low frequency approximation for a class of linear quantum
  systems using cascade cavity realization,'' \emph{Systems $\&$ Control
  Letters}, vol.~61, no.~1, pp. 173--179, 2012.

\bibitem{GY16:tac}
M.~Gu{\c{t}}{\u{a}} and N.~Yamamoto, ``System identification for passive linear
  quantum systems,'' \emph{IEEE Transactions on Automatic Control}, vol.~61,
  no.~4, pp. 921--936, 2016.

\bibitem{GZ15:auto}
J.~E. Gough and G.~Zhang, ``On realization theory of quantum linear systems,''
  \emph{Automatica}, vol.~59, pp. 139--151, 2015.

\bibitem{ZJK96:book}
K.~Zhou, J.~C. Doyle, and K.~Glover, \emph{Robust and Optimal Control}.\hskip
  1em plus 0.5em minus 0.4em\relax Prentice Hall, 1996.

\bibitem{GJ09:tac}
J.~Gough and M.~R. James, ``The series product and its application to quantum
  feedforward and feedback networks,'' \emph{IEEE Transactions on Automatic
  Control}, vol.~54, no.~11, pp. 2530--2544, 2009.

\bibitem{IYYYF12:tac}
S.~Iida, M.~Yukawa, H.Yonezawa, N.~Yamamoto, and A.~Furusawa, ``Experimental
  demonstration of coherent feedback control on optical field squeezing,''
  \emph{IEEE Transactions on Automatic Control}, vol.~57, no.~8, pp.
  2045--2050, 2012.

\bibitem{NJD09:siamjco}
H.~I. Nurdin, M.~R. James, and A.~C. Doherty, ``Network synthesis of linear
  dynamical quantum stochastic systems,'' \emph{SIAM Journal on Control and
  Optimization}, vol.~48, no.~4, pp. 2686--2718, 2009.

\bibitem{AWSCPM15:arxiv}
\BIBentryALTinterwordspacing
R.~N. Alexander, P.~Wang, N.~Sridhar, M.~Chen, O.~Pfister, and N.~C. Menicucci,
  ``One-way quantum computing with arbitrarily large time-frequency
  continuous-variable cluster states from a single optical parametric
  oscillator,'' 2015. [Online]. Available:
  \url{http://arxiv.org/abs/1509.00484}
\BIBentrySTDinterwordspacing

\bibitem{MPC11:pra}
C.~A. Muschik, E.~S. Polzik, and J.~I. Cirac, ``Dissipatively driven
  entanglement of two macroscopic atomic ensembles,'' \emph{Physical Review A},
  vol.~83, no.~5, p. 052312, 2011.

\bibitem{GK04:book}
C.~Gerry and P.~Knight, \emph{Introductory Quantum Optics}.\hskip 1em plus
  0.5em minus 0.4em\relax Cambridge University Press, 2004.

\end{thebibliography}
\end{document}